\documentclass[american,aps,pra,reprint,floatfix,nofootinbib,superscriptaddress]{revtex4-1}

\usepackage{graphicx}
\usepackage{color}
\usepackage{amsfonts}
\usepackage{amsmath}
\usepackage{tikz}
\usepackage{times}
\usepackage{hyperref}

\newtheorem{theorem}{Theorem}

\newenvironment{proof}[1][Proof]{\noindent\textbf{#1.} }{\ \rule{0.5em}{0.5em}}
\def\Tr{\operatorname{Tr}}

\newcommand{\bra}[1]{\ensuremath{\left\langle#1\right|}}
\newcommand{\ket}[1]{\ensuremath{\left|#1\right\rangle}}

\let\originalleft\left
\let\originalright\right
\renewcommand{\left}{\mathopen{}\mathclose\bgroup\originalleft}
\renewcommand{\right}{\aftergroup\egroup\originalright}

\begin{document}
	\title{ Monogamy of Quantum Privacy}
	
	\author{Arun Kumar Pati}
	\email{akpati@hri.res.in}
	\affiliation{Quantum Information and Computation Group,\\
	Harish-Chandra Research Institute, HBNI, Chhatnag Road, Jhunsi, Allahabad 211019, Uttar Pradesh, India}

	\author{Kratveer Singh}
	\affiliation{Indian Institute of Science Education and Research Bhopal, Bhopal Bypass Road, Bhauri, Bhopal 462066, Madhya Pradesh, India}
	
	\author{Manish K. Gupta}
	\email{manishgupta@hri.res.in}
	\affiliation{Quantum Information and Computation Group,\\
		Harish-Chandra Research Institute, HBNI, Chhatnag Road, Jhunsi, Allahabad 211019, Uttar Pradesh, India}

	\date{\today }
	
	\begin{abstract}
	Quantum mechanics ensures that the information stored in a quantum state 
	is secure and the ability to send private information through a quantum channel is at least as great as the coherent information.~We derive trade-off relations between quantum privacy, information gain by Eve and the disturbance caused by Eve to the quantum state that is being sent through a noisy channel. For tripartite quantum states, we show that monogamy of privacy exists in the case of a single sender and multiple receivers. When Alice prepares a tripartite entangled state and shares it with Bob and Charlie through two different noisy quantum channels, we prove that if the minimally guaranteed quantum privacy between Alice and Bob is positive, then the privacy of information between Alice and Charlie has to be negative. Thus, quantum privacy for more than two parties respects mutual exclusiveness. Then, we prove a monogamy relation for the minimally guaranteed quantum privacy for tripartite systems. We also prove a trade-off relation between the entanglement of formation across one partition and the quantum privacy along another partition. Our results show that quantum privacy cannot be freely shared among multiple parties and can have implication in future quantum networks.
	\end{abstract}
	
	
	\maketitle
	
	\section{Introduction}
	Quantum mechanics ensures the privacy of a quantum state.  Given a single copy of the state, it is impossible to determine its quantum state. This is often corroborated by the no-cloning theorem \cite{noCloning82}. Similarly, if we tamper with a quantum state and try to delete it, the information moves from the original system to the environment as captured by the no-deletion theorem \cite{noDeletion}. Even more generally, if a physical process attempts to take away information from a quantum system and hides it in the correlation, then it is impossible to do so. The no-hiding theorem \cite{PhysRevLett.98.080502} proves that information which appears to be missing from the original system remains in the environment from which it can be recovered, in principle \cite{PhysRevLett.106.080401}. All these results have implications for privacy and entail potence of quantum information in the presence of environment.

	A quantum communication channel can be used for various quantum information processing tasks such as the transmission of quantum and classical information as may happen in the case of super dense coding \cite{PhysRevLett.69.2881}, teleportation \cite{PhysRevLett.70.1895}, remote state preparation \cite{PhysRevA.63.014302}, and distributed quantum dense coding \cite{PhysRevA.87.052319,PhysRevLett.93.210501,DC-2006,PhysRevA.92.052330}. It can also be used to share the information between a sender and a receiver that is reliably secret from any third party and can be used as a cryptographic key for private communication with the sender \cite{BB84,PhysRevLett.67.661}.
	
	Quantum communication channel has a distinct advantage over the classical communication channel. If there is information leakage to an eavesdropper (Eve) that is trying to infer the communication happening between the sender (Alice) and the receiver (Bob), then her presence can be detected. The secrecy of quantum communication against copying of the signal by Eve is guaranteed by laws of physics as it was shown in the quantum key distribution (QKD) protocol independently by Bennet and Brassard, and by Ekert \cite{BB84,PhysRevLett.67.661}. Sharing of the secret key is important for providing privacy as one can use it in cryptographic protocols such as the one-time pad \cite{OneTimePad}. Although quantum cryptography provides the best security available at present, it is not immune to attack and exploits by Eve due 
	to leakage of information or due to botched implementations \cite{PhysRevA.75.032314,QKD-Attack-2010,NJP-2011}. The existing systems may be prone to side-channel attacks that rely on imperfect experimental implementation and hence side-channel free QKD was proposed that replaces real channels with virtual channels in QKD protocol to eliminate the attack \cite{PhysRevLett.108.130502}. Further study on absolute limits of privacy has been done by Ekert and Renner \cite{Ekert-2014,RennerThesis05} and a recent security proof for quantum key distribution has been carried out by Tomamichel and Leverrier \cite{Tomamichel2017largelyself}. 
	
    In realistic quantum communication, the channel is always noisy and the information that is leaked to the environment reveals the activity between the sender and the receiver to the eavesdropper. For a single sender and a single receiver, it was initially understood that the critical part of such communication task is to securely share entanglement between 
    two parties. The amount of information that can be securely shared is proportional to the amount of entanglement that can be shared between the two parties. It was shown that secret key sharing between two parties is equivalent to ``entanglement purification" \cite{PhysRevLett.77.2818}. Later it was found that bound entangled states can also be used to share secret key without third-party sharing it. Security of such scheme was further investigated in Ref.\cite{PhysRevLett.79.4034,PhysRevA.54.2675,PhysRevLett.78.2256}.  Schumacher and Westmoreland quantified the privacy of a channel that is measured by the information available to the receiver and not available to any eavesdropper and showed that it can be made as large as the channel's coherent information \cite{PhysRevLett.80.5695,PhysRevA.54.2629}.
	
 	In this paper, we ask the question if quantum mechanics can ensure privacy for more than two parties? To answer this question we generalize the result of Schumacher and Westmoreland \cite{PhysRevLett.80.5695} and analyze the case where a sender (Alice) shares the entangled state with two receivers (Bob and Charlie) over a noisy quantum channel. First, we prove trade-off relations between quantum privacy, information gain by Eve and the disturbance caused by Eve to the quantum state that is being sent through a noisy channel. Next, we show that the minimal guaranteed quantum privacy obeys a strict monogamy relation for a single sender and two receivers. For a tripartite entangled state shared between Alice, Bob and Charlie through two different noisy quantum channels, we prove that if the quantum privacy between Alice and Bob is positive then the privacy of information between Alice and Charlie has to be negative. Thus, quantum privacy for more than two parties respects mutual exclusiveness. This means that if Alice and Bob can reliably share a secret key, then Alice and Charlie cannot share a secret key. Then, we prove a monogamy relation for the minimally guaranteed quantum privacy for tripartite systems. We also prove a trade-off relation between the entanglement of formation across one partition and the quantum privacy along another partition.
	
   The paper is organized as follows. In Section \ref{privacy}, we first review the definition of coherent information and then recall the result by 
   Schumacher and Westmoreland \cite{PhysRevLett.80.5695} on the privacy of a single sender and single receiver. In section \ref{trade}, we prove a trade-off relation between quantum privacy and information gain by the eavesdropper.  Similarly, we prove a trade-off relation between quantum privacy and the disturbance caused to the quantum state due to the presence of the eavesdropper. In section \ref{monogamy-relations}, we prove the mutually exclusive relation between the quantum privacy for a single sender and two receivers. Then, we show the monogamy relations for
   the privacy of a single sender and two receivers. Also, we find a complementarity relation between the entanglement of formation across one partition and the quantum privacy across another partition.  Finally, in section \ref{con}, we discuss the implications of this monogamy relation in quantum networks.
   	
\section{Coherent Information and Privacy}
	\label{privacy}
	We imagine a situation where Alice wants to send quantum information to Bob over a noisy quantum channel. She prepares a quantum system $\rho^{Q}$ that is part of a larger system $\rho^{RQ} = \ket{\Psi^{RQ} }\bra{\Psi^{RQ}} $ such that this compound system is initially in a pure entangled state. The system $R$ is called \emph{reference system} and $\rho^{RQ}$ is called ``purification" of $\rho^{Q}$ with $\rho^{Q}= \Tr_{R}( \ket{\Psi^{RQ} } \bra{\Psi^{RQ} } )$.

		\begin{figure}[ht!]
		\begin{center}
			\begin{tikzpicture}
			\draw[gray,very thick, dashed] (-2.5,-0.5) -- (-1.5,0.5) (-3,-0.5) node[black]{Alice};
			\draw [gray,very thick, dashed] (-1.5,0.5) -- (2,0.5) (-0.75,0.75) node[black]{R};
			\draw[gray,very thick, dashed] (2.0,0.5) -- (3,-0.5);
			\draw[black, very thick] (-1.5,-1.5) -- (0.0,-1.5);
			\draw[black, very thick] (-2.5,-0.5) -- (-1.5,-1.5) (-3,-0.9) node[black]{$\ensuremath{|\Psi^{RQ}\rangle}$};
			\draw[black,very thick] (0.0,-2.0) rectangle (0.8,-1) (0.4,-1.5) node[inner sep=0pt] {$U_{\mathcal{N}}$} (-0.75,-1.25) node[]{$Q$}; 
			\draw[black, very thick] (0.8,-1.25) -- (2,-1.25) (3.4,-0.5) node[black]{Bob} (3.4,-1) node[]{$\ensuremath{|\Psi^{RQ'E'}\rangle}$};
			\draw[black, very thick] (0.8,-1.75) -- (2,-1.75) (2.4,-1.75) node[black]{Eve} (1.5,-1.5) node[]{$E'$};
			\draw[black, very thick] (2,-1.25) -- (3,-0.5) (1.5,-1.0) node[]{$Q'$};
			\end{tikzpicture}
			\caption{Alice and Bob share the entanglement through a noisy channel. The joint state that Alice and Bob hold with the environment is 
			$\ensuremath{|\Psi^{RQ'E'}\rangle}$. $U_{\mathcal{N}}$ realizes the noisy channel as a unitary evolution in the enlarged Hilbert space.}
			\label{fig:one-receiver}
		\end{center}
	\end{figure}
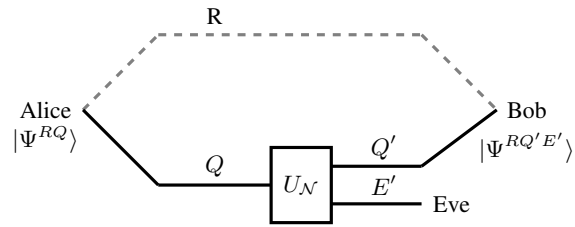

    Alice sends the state $Q$ over the noisy quantum channel to Bob. The noisy evolution of the state is described by the action of superoperator $\mathcal{E}^{Q}$ on the state $\rho^{Q}$, which is a completely positive trace preserving (CPTP) map. This noisy evolution described by superoperator can be represented by a unitary evolution operator $U^{QE}$ on a larger system that includes the environment, which is assumed to be initially in pure state $\ket{E}$. The final state of the system $Q$ received by Bob is given by~
    \begin{align}
	\rho^{Q'} &= \mathcal{E}^{Q} \left(\rho^{Q}\right) \nonumber \\
	&= \Tr_{E}\left[ \left(I^{R} \otimes U^{QE}\right)\left(\rho^{RQ} \otimes \ket{E}\bra{E} \right)\left(I^{R} \otimes{U^{QE}}^{\dagger}\right) \right].
	\end{align}
	As discussed in Ref.\cite{PhysRevA.54.2614}, the amount of information that is exchanged between the system $Q$ and the environment $E$ during the interaction is measured by the von Neumann entropy $S_{e}$. Since the environment is initially in a pure state, the entropy exchange is $S_{e}=S(\rho^{E'})$. The \emph{coherent quantum information}, 
    $I_{c}$ as introduced in Ref.\cite{PhysRevA.54.2629}, is given by~
	\begin{align}
	I_c(R \rangle Q')  = S(\rho^{Q'})-S_{e}
	    = S(\rho^{Q'})-S(\rho^{RQ'}).
	\end{align}
	In this paper, we will use the notation $I_c= I_c(R \rangle Q')= I_c(A \rangle B)$, as the subsystem $R$ is with Alice and the subsystem $Q'$ is with Bob, after $Q$ passes through a noisy channel. The coherent information captures how much entanglement can be retained between Alice and Bob when Alice sends one-half of an entangled pair through a noisy channel. The notion of coherent information plays an important role in quantum data processing and quantum error correction. It is an intrinsic quantity and satisfies the following properties \cite{wilde_2013}: 
	(i) the absolute value of the coherent information obeys $|I_c(R \rangle Q')| \le \log \dim {\cal H}_R $, (ii) under quantum operation it can never increase, i.e., it
	satisfies the data processing inequality $I_c(R \rangle Q') \leq I_c(R \rangle Q)$ as $\rho^{RQ'}=I^{R} \otimes \mathcal{E}^{Q}\left( \rho^{RQ}\right)$, and 
	(iii) $S(\rho^Q) \ge I_c(R \rangle Q')$. Note that iff $ S(\rho^Q) = I_c(R \rangle Q')$, then perfect quantum error-correction is possible.~Schumacher and Westmoreland \cite{PhysRevLett.80.5695} have shown that the ``optimal guaranteed privacy'' of the communication channel between Alice and Bob, as depicted in Fig.\ref{fig:one-receiver}, is lower bounded by the coherent information $I_c(A \rangle B)$.
	
	Next, imagine a situation where Alice wants to send classical information to Bob and she prepares a quantum system $Q$ by encoding information in one of the ``signal states", $\rho_{k}^{Q}$ with {\it a priori} probability $p_{k}$ such that the average state $\rho^{Q}$ is given by~
	\begin{align}
	\rho^{Q} = \sum_{k} p_{k}\rho_{k}^{Q}.
	\end{align}
	She sends the state over a noisy quantum channel and Bob receives the $k_{th}$ signal state as $\rho_{k}^{Q'}= \mathcal{E}^{Q}(\rho_{k}^{Q} )$. Since the superoperator is linear, the average state received by Bob is given by~
	\begin{align}
	\rho^{Q'} &= \sum_{k} p_{k} \mathcal{E}^{Q}(\rho_{k}^{Q}) = \mathcal{E}^{Q} \left(\rho^{Q}\right).
	\end{align}
	The amount of classical information $H_{Bob}$ that can be conveyed from Alice to Bob is governed by the Holevo quantity $\chi^{Q'}$, where~
	\begin{align}
	\chi^{Q'} = S(\rho^{Q'}) -\sum_k p_k  S(\rho_k^{Q'}).
	\end{align}
	
	The evolution superoperator $\mathcal{E}^{Q}$ that represents the effect of the eavesdropper can be represented as a unitary $U^{QE}$ evolution of a larger quantum system that includes an environment $E$. The evolution can be shown to be~
	\begin{align}
	\rho^{RQ} \otimes \ket{E}\bra{E} \xrightarrow[]{U^{QE}} \rho^{RQ'E'}.
	\end{align}
	The  amount of classical information $H_{Eve}$ that is available to Eve is governed by the Holevo quantity $\chi^{E'}$, where~
	\begin{align}
	\chi^{E'} = S(\rho^{E'}) -\sum_k p_k  S(\rho_k^{E'}).
	\end{align}
	
	The quantum ``privacy" of a channel between Alice and Bob is defined as~
	\begin{align}
	P_{AB} &= H_{Bob}-H_{Eve}.
	\end{align}
	Classically, any positive difference between $(H_{Bob}-H_{Eve})$ can be used to create a reliably secret key string of length $P$ \cite{1055892,256484,Devetak207}. We know that $H_{Bob} \leq \chi^{Q'} $, since Alice and Bob can use the best possible strategy to maximize the information gain from the channel. Similarly, for the eavesdropper,
	we know that $H_{Eve} \leq \chi^{E'}$ as Eve will also apply the best possible strategy to maximize the information gain that has been leaked to the environment. 
	One can define the guaranteed privacy $P_G = {\rm inf} P_{AB}$, where the infimum is taken over all the possible strategies that Eve can implement. The optimal guaranteed 
	privacy ${\cal P}_{AB} $ is defined as ${\cal P}_{AB} = {\rm sup} P_G$, where the supremum is taken over all strategies that Alice and Bob can implement. With this definition, Schumacher and Westmoreland \cite{PhysRevLett.80.5695} argued that~
	\begin{align}
	{\cal P}_{AB} \geq \chi^{Q'} - \chi^{E'} = I_c(A \rangle B),
	\end{align}
	where the coherent information $I_c(A\rangle B) = (\chi^{Q'} - \chi^{E'}) $. Therefore, the optimal	guaranteed \emph{private} information ${\cal P}_{AB}$ that can be shared between Alice and Bob is at least as great as its ability to send \emph{coherent} information.

\section{Trade-off Relations for Quantum Privacy}
\label{trade}

	In this section, we prove trade-off relations between quantum privacy, information gain by Eve and the disturbance caused to the quantum state by Eve during the eavesdropping. Intuitively, one can imagine that there should be some complementarity relation between quantum privacy and information gain by Eve. Indeed, we can find a trade-off relation between the amount of quantum information sent from Alice to Bob and the amount of classical information that Eve can gain. This is given by~
	\begin{align}
	I_c(A \rangle B) + H_{Eve} \le \log d,
	\end{align}
	where $d= dim {\cal H}_Q$. Physically, this shows that if channel is able to send more distinct quantum information from Alice to Bob, the amount of information gained by Eve will be less and vice versa. This is logically understandable. In the ideal situation, if there is no eavesdropper then perfect quantum information can be transmitted from Alice to Bob.

	Similarly, we can obtain a trade-off relation between the quantum privacy and the amount of classical information that Eve can gain, i.e.,~
	\begin{align}
	P_{AB} + H_{Eve} \le \log d.
	\end{align}
	This shows that if Alice and Bob can create more reliably secret string of key, then Eve will obtain less amount of information from the communication channel.
	
	 Suppose, we give all the computational power to Eve, then how much privacy can be maintained between Alice and Bob? In a worst-case scenario, Eve will try to make $H_{Eve}$ as close as possible to $\chi^{E'}$ by quantum technological arsenal at her disposal. Therefore, we can define ``minimal guaranteed privacy'' between Alice and Bob as~
	\begin{align}
	 P_{AB}^{min} =H_{Bob} -  {\rm max} H_{Eve} =  H_{Bob}- \chi^{E'}.
	\end{align}
	Since $H_{Bob} \le \chi^{Q'}$, we have $P_{AB}^{min} \le \chi^{Q'} - \chi^{E'}$ and hence $ P_{AB}^{min} \le I_c(A \rangle B)$. Thus, the minimal guaranteed privacy and the optimal guaranteed privacy obey the inequality~
	\begin{align}
	P_{AB}^{min} \le I_c(A \rangle B) \le {\cal P}_{AB}.
	\end{align}
	
	Now, one may ask how small can be the minimal guaranteed privacy? We can show that the minimal guaranteed privacy can be as small as the entanglement of formation \cite{PhysRevLett.76.722,PhysRevA.54.3824} for $\rho^{RQ'}$. Using the Carlen-Lieb inequality $E_{f} \left( \rho^{RQ'} \right) \geq {\rm max}  \left\lbrace S\left(R\right)-S\left(RQ'\right),S\left(Q'\right)-S\left(RQ'\right), 0 \right\rbrace$ \cite{CarlenLieb13} and assuming that $S\left(Q'\right)-S\left(RQ'\right)$ is the maximum of right-hand side, we have $P^{min}_{AB}~\leq~E_{f}\left(RQ'\right)$. 

	Next, we ask if there can be a trade-off relation between the quantum privacy and the disturbance caused by Eve to the quantum state that is sent through a noisy channel. We show that there is indeed a trade-off relation between them. Intuitively, one may say that a system is disturbed when the initial state is different than the final state and it is not possible to go back to the initial state in a reversible manner. The disturbance is usually an irreversible change in the state of the system, caused by a quantum channel. The disturbance measure should satisfy the following conditions: $(i)$ $D$ should be a function of the initial state $\rho^Q$ and the CPTP $\mathcal{E}^{Q}$ only, $(ii)$~$D$ should be null iff the CPTP map is invertible on the initial state $\rho^Q$, because in this case the change in state can be reversed hence the system is not disturbed, $(iii)$ $D$ should be monotonically non-decreasing under successive application of CPTP maps and $(iv)$ $D$ should be continuous for maps and initial states which do not differ too much.
	Keeping in mind these desirable properties, Maccone \cite{Maccone07} has defined a measure of disturbance that satisfies the above conditions. This is given by~
	\begin{align}
	D(\rho^Q, \mathcal{E}^{Q}) & \equiv S(\rho^Q) -I_c(R \rangle Q') \nonumber \\ 
	&= S(\rho^Q) - S(\mathcal{E}(\rho^Q)) + S( I \otimes (\mathcal{E}) (\ket{\Psi}^{RQ}\bra{\Psi})),
	\end{align}
	where $I_c(R \rangle Q')$ is the coherent information for $\rho^Q$ through a channel ${\cal E}^Q$. Since the coherent information $I_c$ is non-increasing under successive application of quantum operation, the disturbance measure is  indeed monotonically non-decreasing under CPTP map.  Note that $D(\rho^Q,\mathcal{E}^{Q})$ satisfies $0~\leq~D(\rho^Q, \mathcal{E}^{Q}) \leq 2\log_2({\rm dim}{\cal H}_Q)$.
	
	The trade-off relation between the disturbance and the minimal guaranteed privacy is given by~
	\begin{align}
	D(\rho^Q, \mathcal{E}^{Q}) + P_{AB}^{min} \le S(\rho^Q). 
	\end{align}
	Thus, given the limited amount of local entropy $S(A) = S(\rho^Q) $, the amount of minimal guaranteed privacy cannot be more if the disturbance caused by Eve is large. Since $S\left( A\right) = E \left( \ket{\Psi^{RQ}} \right) = -\Tr{\left(\rho^{Q}\log{\rho^{Q}}\right)}$ is the initial entanglement between $R$ and $Q$, the disturbance and the minimal privacy is also bounded by the  initial entanglement.
  
\section{Monogamy of Privacy}
\label{monogamy-relations}
    We now come to the main part of the paper where we analyze the private communication between a single sender (Alice) and two receivers (Bob and Charlie) as shown in Fig.\ref{fig:two-receiver}, where Alice shares the entangled state with Bob and Charlie by sending parts of the system over two separate noisy channels to them.

	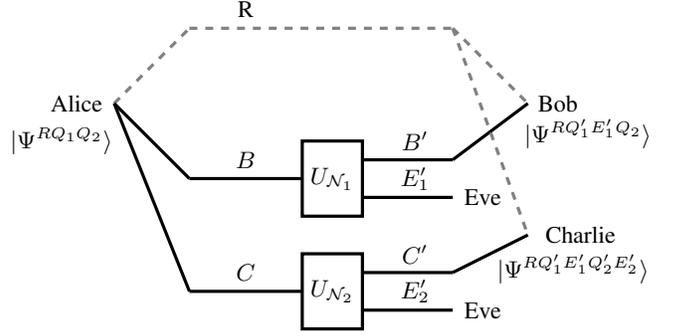
\begin{figure}[ht!]
		\begin{center}
			\begin{tikzpicture}
			\draw[gray,very thick, dashed] (-2.5,-0.5) -- (-1.5,0.5) (-3,-0.5) node[black]{Alice};
			\draw [gray,very thick, dashed] (-1.5,0.5) -- (2,0.5) (-0.75,0.75) node[black]{R};
			\draw[gray,very thick, dashed] (2.0,0.5) -- (3,-0.5);
			\draw[black, very thick] (-1.5,-1.5) -- (0.0,-1.5);
			\draw[black, very thick] (-2.5,-0.5) -- (-1.5,-1.5) (-3,-0.9) (-3.2,-1.0) node[black]{$\ensuremath{|\Psi^{RQ_{1}Q_{2}}\rangle}$};
			\draw[black,very thick] (0.0,-2.0) rectangle (0.8,-1) (0.4,-1.5) node[inner sep=0pt] {$U_{\mathcal{N}_{1}}$} (-0.75,-1.25) node[]{$B$}; 
			\draw[black, very thick] (0.8,-1.25) -- (2,-1.25) (3.4,-0.5) node[black]{Bob} (3.8,-0.9) node[black]{$\ensuremath{|\Psi^{RQ'_{1}E'_{1}Q_{2}}\rangle}$};
			\draw[black, very thick] (0.8,-1.75) -- (2,-1.75) (2.4,-1.75) node[black]{Eve} (1.5,-1.5) node[]{$E'_{1}$};
			\draw[black, very thick] (2,-1.25) -- (3,-0.5) (1.5,-1.0) node[]{$B'$};
			\draw[black, very thick] (-2.5,-0.5) -- (-1.5,-3.0) (-0.75,-2.75) node[black]{$C$};
			\draw[black, very thick] (-1.5,-3.0) -- (0.0,-3.0);
			\draw[black,very thick] (0.0,-2.5) rectangle (0.8,-3.5) (0.4,-3.0) node[inner sep=0pt,black] {$U_{\mathcal{N}_{2}}$} (1.5,-2.50) node[black]{$C'$}; 
			\draw[black, very thick] (0.8,-2.75) -- (2,-2.75) (3.7,-2.25) node[black]{Charlie} (3.6,-2.7) node[black]{$\ensuremath{|\Psi^{RQ'_{1}E'_{1}Q'_{2}E'_{2}}\rangle}$};
			\draw[black, very thick] (0.8,-3.25) -- (2,-3.25) (2.4,-3.25) node[black]{Eve} (1.5,-3.0) node[black]{$E'_{2}$};
			\draw[gray,very thick, dashed] (2.0,0.5) -- (3,-2.25);
			\draw[black,very thick] (2,-2.75) -- (3,-2.25);
			\end{tikzpicture}
			\caption{Alice, Bob, and Charlie share the entanglement through the noisy channels. The joint state after passing through two different noisy channels are  
			$\ensuremath{|\Psi^{RQ'_{1}E'_{1}Q_{2}}\rangle}$ and  $\ensuremath{|\Psi^{RQ'_{1}E'_{1}Q'_{2}E'_{2}}\rangle}$. $U_{\mathcal{N}_{1}}$ and $U_{\mathcal{N}_{2}}$ are the unitary realization of the noisy channels $\mathcal{E}_{1}^{Q}$ and $\mathcal{E}_{2}^{Q}$, respectively.}
			\label{fig:two-receiver}
		\end{center}
	\end{figure}

	In the sequel, we prove the monogamy relation for quantum privacy for the tripartite quantum system and explore how entanglement across one partition affects the privacy across another partition. Imagine that Alice prepares a pure entangled state $\rho^{RQ_{1}Q_{2}}=\ket{\Psi^{RQ_{1}Q_{2}}}\bra{\Psi^{RQ_{1}Q_{2}}}$ and sends $Q_{1}$ part of the system to Bob and $Q_{2}$ part to Charlie over noisy channels. The noisy channels act on the state and their action may be described by the superoperators $\mathcal{E}_{1}^{Q}~\text{and}~ \mathcal{E}_{2}^{Q}$, respectively. The superoperators $\mathcal{E}_{1}^{Q}~\text{and}~ \mathcal{E}_{2}^{Q}$ can be represented as unitary evolutions of larger quantum systems that include environments $E_{1}~\text{and}~E_{2}$, respectively. The environments can be assumed to be initially in pure state $\ket{E_{1}}$ and $\ket{E_{2}}$, respectively. The evolution can be shown to be~
	\begin{align}
	\rho^{RQ_{1}Q_{2}} \otimes \ket{E_{1}}\bra{E_{1}} &\xrightarrow[]{U_{Q_{1}E_{1}}} \rho^{RQ'_{1} E'_{1} Q_{2} } \\
	\rho^{RQ'_{1}E'_{1}Q_{2}} \otimes \ket{E_{2}}\bra{E_{2}} &\xrightarrow[]{U_{Q_{2}E_{2}}} \rho^{RQ'_{1}E'_{1}Q'_{2}E'_{2}}.
	\end{align}
	
    Since local operations do not affect entanglement of the parties involved, hence $RQ_{2}$ acts as a reference system during the communication between Alice to Bob, and remains unchanged. Hence, we can define the privacy of the channel between Alice to Bob as~
	\begin{align}
	P_{AB} = H_{Bob}-H_{Eve_1}.
	\end{align}
	Similarly, $RQ'_{1}E'_{1}$ acts as the reference system during the communication between Alice to Charlie, and remains unchanged. Hence, we can define the privacy of channel between Alice to Charlie as~
	\begin{align}
	P_{AC} &= H_{Charlie}-H_{Eve_2}.
	\end{align}

	Now, we ask if we give all the computational power to Eve, then how much privacy can be maintained between Alice to Bob and Alice to Charlie.
	In a competitive scenario, Bob will want to make $H_{Bob}$ as close as possible to $\chi^{Q'_1}$ by a suitable choice of coding and decoding observable. Similarly, Eve will try to make $H_{Eve_1}$ as close as possible to $\chi^{E'_{1}}$ by quantum technological arsenal at her disposal. Therefore, we can define ``minimal guaranteed privacy'' between Alice and Bob as~
	\begin{align}
	 P_{AB}^{min} =H_{Bob} -  {\rm max} H_{Eve{_1}}=  H_{Bob}- \chi^{E'_{1}}.
	\end{align}
	Since $H_{Bob} \le \chi^{Q'_1}$, we have $P_{AB}^{min} \le \chi^{Q'_1} - \chi^{E'_{1}}$ and hence $ P_{AB}^{min} \le I_c(A \rangle B)$. 
	Thus, the minimal guaranteed privacy and the optimal guaranteed privacy across Alice and Bob obey the inequality $ P_{AB}^{min} \le I_c(A \rangle B) \le {\cal P}_{AB}$.

	Similarly, we can define ``minimal guaranteed privacy'' between Alice and Charlie as~
	\begin{align}
	P_{AC}^{min} = H_{Bob} - {\rm max} H_{Eve{_2}} =  H_{Bob}- \chi^{E'_{2}}.
	\end{align}
	The minimal guaranteed privacy and the optimal guaranteed privacy across Alice and Charlie obey the inequality $ P_{AC}^{min} \le I_c(A \rangle C) \le {\cal P}_{AC}$.

	Having defined the minimal quantum privacy for the communication channel between Alice to Bob and between Alice to Charlie, we prove that they obey the exclusive monogamy inequality.
	 \begin{theorem}
	If $P_{AB}^{min}$ and $P_{AC}^{min}$ are the minimal privacy of information between Alice to Bob and Alice to Charlie, respectively, then the following mutually exclusive relation holds
	\begin{equation}
	P_{AB}^{min} + P_{AC}^{min} \le 0 .
	\end{equation}
	\end{theorem}
	
	\begin{proof}
	By noting that $P_{AB}^{min} + P_{AC}^{min} \le (\chi^{Q'_{1}} - \chi^{E'_{1}}) + (\chi^{Q'_{2}} - \chi^{E'_{2}})$, we have
	 $P_{AB}^{min} + P_{AC}^{min} \le I_c(A \rangle B) + I_c(A \rangle C)$. 
	 Using the strong subadditivity of the von Neumann entropy \cite{Ruskai1973,PhysRevA.87.052319} for the tripartite system $ABC$ we have 
	\begin{align}
	\label{strongSubadditivityEq}
	S\left( \rho^{Q'_1} \right) + S\left( \rho^{Q'_2} \right) - S\left( \rho^{RQ'_1} \right) - S\left( \rho^{RQ'_2} \right) \leq 0.
	\end{align}
	Rearranging the above inequality we can write (\ref{strongSubadditivityEq}) as
	\begin{align}
	I_c\left(A \rangle B\right) + I_c\left(A \rangle C\right) &\leq 0. 
	\end{align}
	Hence, we have
	\begin{align}
	P_{AB}^{min} +P_{AC}^{min} &\leq 0.
	\end{align}
	\end{proof}

	This shows that Alice cannot share privacy simultaneously with Bob and Charlie. If the minimal quantum privacy of information between Alice and Bob is positive then the minimal privacy of information between Alice and Charlie has to be negative and vice versa. This means that if Alice shares a secret string of key with Bob, then she cannot share the secret key with Charlie.
	
	The above result for the single sender and two receivers can be generalized to the single sender and multiple receivers. Suppose we have a multipartite entangled state $\rho^{ABC\cdots N}$ shared between Alice, Bob, Charlie,$\cdots$, and Nancy. Now, we define the minimum privacy across $AB,AC,\cdots, AN$ over different noisy channels as $P^{min}_{AB},P^{min}_{AC},\cdots, \text{and}~P^{min}_{AN}$, respectively. We know that
	\begin{align}
	\sum_{i=1}^{N-1} S\left(A|B_{i}\right) \geq 0,
	\end{align}
	where $B_{i}=\lbrace B,C,D,\cdots,N\rbrace$.~This inequality can be proved by combining the strong subadditivity of the von Neumann entropy and the subadditivity of the conditional entropy \cite{PhysRevA.87.022314,nielsen_chuang_2010}. Using $P_{AB_{i}}^{min}  \leq I_c\left(A \rangle B_{i}\right)$ and negating the above inequality we can easily show that 
	\begin{align}
	\sum_{i=1}^{N-1} P_{AB_{i}}^{min}  \leq 0.
	\end{align}
	This means that Alice cannot share privacy simultaneously with Bob, Charlie,$\cdots$, and Nancy. If the minimal quantum privacy of information between Alice and Bob is positive then the minimal privacy of information between Alice and others has to be negative and vice versa. 
	
	\begin{theorem}
	If $P_{AB}^{min}$ and $P_{AC}^{min}$ defines the minimal guaranteed privacy of information between Alice-Bob and Alice-Charlie respectively, and ${\cal P}_{A(BC)}$ is the optimal guaranteed privacy between Alice to Bob and Alice, then the following monogamy relation holds
	\begin{equation}
	P_{AB}^{min} + P_{AC}^{min} \le {\cal P}_{A(BC)},
	\end{equation}
	where ${\cal P}_{A(BC)}$ is defined as ${\cal P}_{A(BC)} = sup P_G$ with $P_G \ge H_{Bob+Charlie} - \chi^{E}$ such that Bob and Charlie are allowed to do a joint measurement to maximize the information gain.
	\end{theorem}
		
	\begin{proof}
	Since $P_{AB}^{min} + P_{AC}^{min} \le I_c(A \rangle B) + I_c(A \rangle C)$, we can use the strong subadditivity of entropy for the tripartite case \cite{nielsen_chuang_2010}, i.e.,
	
	\begin{align}
	S\left(\rho^{RQ'_1Q'_2} \right) + S\left(\rho^{Q'_1}\right) + S\left( \rho^{Q'_2}\right) \nonumber \\
	\leq S\left(\rho^{RQ'_1}\right) + 
	S\left(\rho^{RQ'_2}\right) + S\left( \rho^{Q'_1Q'_2}\right).
	\end{align}
	This is equivalent to  
	\begin{align}
	I_c\left(A\rangle B\right) + I_c\left(A \rangle C\right) &\leq I_c\left(A \rangle BC\right), \\
	{\rm where},~~~~ I_c\left(A\rangle BC\right) &= S\left( \rho^{Q'_1Q'_2}\right) - S\left(\rho^{RQ'_1Q'_2} \right). \nonumber \\
	\end{align}
	Using the inequality for the optimal guaranteed privacy, i.e. ${\cal P}_{A(BC)} \ge I_c(A \rangle BC)$, we have 
	$P_{AB}^{min} + P_{AC}^{min} \le {\cal P}_{A(BC)}$.
	\end{proof}
	
	Now, we ask how entanglement across one partition affects the minimal guaranteed privacy across the other partition. One of the remarkable feature in multi-partite scenario is that entanglement often respects monogamy \cite{PhysRevA.61.052306,PhysRevA.62.050302,PhysRevA.65.010301,HIGUCHI2000213}. Though not all measures of entanglement obey monogamy, they do respect it in a qualitative manner, in the sense, that if Alice is maximally entangled with Bob, then she cannot get maximally entangled with Charlie at the same time. For example, the square of concurrence satisfies the monogamy but the entanglement of formation does not satisfy the same. However, it was shown that the square of entanglement of formation does satisfy monogamy inequality for multiqubit states \cite{PhysRevLett.113.100503}. Similarly there are other correlation measures which can be monogamous under certain conditions. For example, quantum discord is not monogamous in general \cite{PhysRevA.85.040102}. Even though the mutual information and the entanglement of purification \cite{JMP.43.4286} capture the total correlation in a bipartite state, the mutual information for tripartite pure state is monogamous but the entanglement of purification is not \cite{PhysRevA.91.042323}. For a recent review see Ref.\cite{monogamyRev17}.

	Koashi and Winter \cite{PhysRevA.69.022309} proved that given a tripartite quantum state shared between Alice, Bob and Charlie, the entanglement 
	cost $E_C(AB)$ across Alice and Bob, and the distillable common randomness $C_D(AC)$ across  Alice and Charlie obey a trade-off relation, i.e.,
	\begin{equation}
	E_C(AB) + C_D(AC) \le S(A),
	\end{equation}
	where $E_C(AB) = {\lim}_{n \rightarrow \infty} \frac{1}{n} E_f({\rho^{RQ_1'}}^{\otimes n})$ 
	is the asymptotic cost to prepare a quantum state from singlets, given by the regularized entanglement of formation 
	and $C_D(AC) = {\lim}_{n \rightarrow \infty} \frac{1}{n} J({\rho^{RQ_2'}}^{\otimes n})$ is the regularized version of the classical correlation. The classical correlation \cite{HendersonVedral01} for a bipartite density operator $\rho^{AB}$ is defined as $J\left(\rho^{RQ'_{2}}\right)=J(A:B)= S(\rho^A) - {\rm min} \sum_i p_i S(\rho^{A|i})$, where $\rho^{A|i}=\frac{\Tr_{B}\left[ \left(I^{A} \otimes \pi^{B}_{i} \right) \rho^{AB} \left(I^{A} \otimes {\pi^{B}_{i}}^{\dagger}\right) \right]}{\Tr_{AB}\left[ \left(I^{A} \otimes \pi^{B}_{i} \right) \rho^{AB} \left(I^{A} \otimes {\pi^{B}_{i}}^{\dagger}\right) \right]}$ is the conditional state of $A$ given that a projective measurement has been performed on the subsystem $B$ with $p_i$ being the probability of having $i$th outcome.

	Here we prove a trade-off relation between the entanglement of formation across $AB$ and the minimal guaranteed privacy across $AC$ for any tripartite state.

	\begin{theorem}
	For any tripartite state, the following relation holds for the entanglement of formation and the minimal guaranteed privacy, i.e., we have
	\begin{equation}
	\label{EfPrivacyEq}
	E_f(AB) + P_{AC}^{min} \le S(A).
	\end{equation}
	
	\end{theorem}
	\begin{proof}
	For any tripartite state $\rho^{RQ'_1 Q'_2}$ we have the Koashi-Winter relation \cite{PhysRevA.69.022309} $E_f(\rho^{RQ'_1}) + J(\rho^{RQ'_2}) \le S(\rho^{R})$, where $J(\rho^{RQ'_2})$ is the classical correlation across Alice and Charlie.
	Using the above relation, we can show that $ E_f(\rho^{RQ'_1}) + I_c(A\rangle C) \le D(\rho^{RQ'_2})$ where $D(\rho^{RQ'_2})$ is the quantum correlation between Alice and Charlie. The quantum correlation $D(\rho^{RQ'_2})$ is nothing but the discord \cite{PhysRevLett.88.017901} across $\rho^{RQ'_{2}}$ and is defined as $D(\rho^{RQ'_2}) = I \left(\rho^{RQ'_{2}}\right)-J\left(\rho^{RQ'_{2}} \right)$, where $ I \left(\rho^{RQ'_{2}}\right) = S\left( \rho^{R}\right) + S\left( \rho^{Q'_{2}}\right) -S\left( \rho^{RQ'_{2}}\right) $ is the mutual information for the state $\rho^{RQ'_{2}}$. 
	
	Now, $ E_f(AB) + P_{AC}^{min} = E_f(\rho^{RQ'_1}) + P_{AC}^{min} \le  E_f(\rho^{RQ'_1}) + I_c(A\rangle C) \le D(\rho^{RQ'_2})$. However, $D(\rho^{RQ'_2}) \le  S(\rho^{R})= S(A) $ \cite{DattaThesis08,PhysRevA.82.052122}.
	Hence the proof. 
	\end{proof}

	This shows that for any tripartite state the amount of entanglement across one partition restricts the quantum privacy that can be shared across another partition. This is reminiscent of the trade-off relation between the entanglement cost and the distillable common randomness for any tripartite system. One can also physically interpret the relation (\ref{EfPrivacyEq}) as follows. From the compression theorem, we know that one can transfer the information contents of the system $R$ into $S(\rho^R)= S\left(A\right)$ qubits per copy. We know that this does retain correlation to other subsystems faithfully in the asymptotically limit. Therefore, the local entropy $S\left(\rho^R\right)$ represents the effective size of the subsystem $R$ measured in qubits. The above trade-off relation shows that the entanglement between the subsystem $R$ and the subsystem $Q'_1$ and the privacy between $R$ and subsystem $Q'_2$ is limited when the size of system $R$ is limited to $S\left(A\right)$ qubits. Physically, this implies that the quantum entanglement between one system and the privacy for the other system are mutually exclusive. Thus, the existence of quantum entanglement across one partition restricts the quantum privacy across the other partition.

\section{Discussions and Conclusions}
\label{con}

	Monogamy of quantum correlation such as quantum entanglement and other correlations play an important role in quantum communication. We have shown that the monogamy of quantum privacy exists for tripartite entangled states as well as for multipartite entangled state. In addition to the monogamy of privacy for the multipartite entangled state one can have a monogamy for the square of the privacy for multiqubit entangled states. To prove this, one can use the monogamy relation for the square of the entanglement of formation \cite{PhysRevLett.113.100503} which is given by
	$E^{2}_{f}\left( \rho_{AB}\right)+E^{2}_{f}\left( \rho_{AC}\right)+ 
	\cdots+E^{2}_{f}\left( \rho_{AN}\right) \leq E^{2}_{f}\left( \rho_{A|BC\cdots N}\right)$.
	For example, if we have a multiqubit entangled state $\rho^{ABC\cdots N}$ shared between Alice, Bob, Charlie,$\cdots$, and Nancy and the minimum privacy across $AB,AC,\cdots, AN$ over different noisy channels as $P^{min}_{AB},P^{min}_{AC},\cdots, \text{and}~P^{min}_{AN}$, respectively, then by using the Carlen-Lieb inequality \cite{CarlenLieb13} we can show that
	${\left(P^{min}_{AB}\right)}^{2}+{\left(P^{min}_{AC}\right)}^{2}+\cdots+ 
	{\left(P^{min}_{AN}\right)}^{2} \leq E^{2}_{f}\left( \rho_{A|BC\cdots N}\right)$.
	Therefore, the square of the minimal privacy for the single sender and multiple receivers respects the monogamy for multiqubit states. 
	
    To summarize, we have proved a trade-off relation between the quantum privacy and the information gain by Eve. Similar trade-off relation holds for the disturbance caused by Eve and the minimal guaranteed quantum privacy.~We have derived the exclusive monogamy relation for quantum privacy in the case of a single sender and multiple receivers and shown that if the sender has a positive privacy of shared information with the first receiver then the privacy with other receiver has to be negative. This means that both the parties cannot reliably generate a secret key due to the restrictions imposed by the exclusive monogamy inequality. Similarly, we proved another monogamy relation for the quantum privacy for any tripartite system which shows that the sum of minimal guaranteed privacies across Alice-Bob and Alice-Charlie are bounded by the optimal guaranteed privacy across Alice-(Bob and Charlie). Thus, quantum privacy cannot be freely shared across many subsystems.~Furthermore, we have proved a trade-off relation between the entanglement of formation and the quantum privacy across different partitions of a tripartite system. We hope that our results will have a significant impact in multiparty quantum key distribution and the quantum network when one is trying to share secret key simultaneously with multiple receivers. From our result, it follows that if Alice wants to share privacy with Bob and Charlie together, it is better to share two separate bipartite entangled states with both the receivers. On the other hand, if Alice wants to share privacy with one of the two receivers at a later point in time, she can share a tripartite entangled state and has the option to choose with whom to share the privacy. In future, it may be worth exploring if the monogamy relation holds for the optimal guaranteed privacy. 
 
	\begin{acknowledgments}
    KS acknowledges the support from HRI, Allahabad for summer visiting program.
	\end{acknowledgments}


\end{document}